\documentclass[a4paper,USenglish,10pt]{article}
\usepackage[utf8]{inputenc}

\usepackage{amsthm,amssymb,amsfonts,amsmath}
\usepackage{amssymb,amsfonts,amsmath}
\newtheorem{theorem}{Theorem}
\newtheorem{lemma}[theorem]{Lemma}
\newtheorem{corollary}[theorem]{Corollary}

\theoremstyle{definition}

\newtheorem{problem}{Problem}

\usepackage[hidelinks]{hyperref}
\usepackage[blocks]{authblk}

\usepackage{graphicx}
\graphicspath{{figures/}}
\usepackage{wrapfig}

\makeatletter
\g@addto@macro\bfseries{\boldmath}
\makeatother

\advance\hoffset-7mm
\usepackage{thm-restate}

\usepackage{url}

\newcommand{\floor}[1]{\ensuremath{\left \lfloor #1 \right \rfloor}}

\newcommand{\A}{\ensuremath{\mathcal{A}}}

\newcommand{\R}{\ensuremath{\mathbb{R}}}

\makeatletter
\g@addto@macro\bfseries{\boldmath}
\makeatother

\usepackage{authblk}

\date{}

\title{Bisecting three classes of lines}

\begin{document}
\author[1]{Alexander Pilz\thanks{Supported by a Schr\"odinger fellowship of the Austrian Science Fund (FWF): J-3847-N35.}}
\affil[1]{Institute of Software Technology, Graz University of Technology.\\ \texttt{apilz@ist.tugraz.at}}

\author[2]{Patrick Schnider}
\affil[2]{Department of Computer Science, ETH Z\"urich.\\ \texttt{patrick.schnider@inf.ethz.ch}}

\maketitle

\begin{abstract}
We consider the following problem: 
Let $\mathcal{L}$ be an arrangement of $n$ lines in $\R^3$ colored red, green, and blue.
Does there exist a vertical plane $P$ such that a line on $P$ simultaneously bisects all three classes of points in the cross-section $\mathcal{L} \cap P$?
Recently, Schnider [SoCG 2019] used topological methods to prove that such a cross-section always exists.
In this work, we give an alternative proof of this fact, using only methods from discrete geometry.
With this combinatorial proof at hand, we devise an $O(n^2\log^2(n))$ time algorithm to find such a plane and the bisector of the induced cross-section.
We do this by providing a general framework, from which we expect that it can be applied to solve similar problems on cross-sections and kinetic points.
\end{abstract}

\section{Introduction}
The ham-sandwich theorem is a fundamental result in discrete and computational geometry.
It asserts that for a set of points in $\R^d$ with $d$ classes, there exists a hyperplane that simultaneously bisects all classes.
This is tight in the sense that in any dimension, there exist sets with $d+1$ classes that cannot be bisected by a hyperplane.
To push the boundary of this limitation further, Luis Barba introduced the following problem.%
\footnote{Luis Barba, personal communication (2018).}

Let $\mathcal{L}$ be an arrangement of $n$ lines in $\R^3$.
The arrangement $\mathcal{L}$ is in \emph{general position} if no two lines are parallel or intersect and no line is vertical (parallel to the $z$-axis).
Let $P$ be a directed plane, not containing any line of $\mathcal{L}$.
A \emph{cross-section} of $\mathcal{L}$ induced by $P$ is the point set defined by the intersection of $\mathcal{L}$ with $P$.
We will usually consider the cross-section to be in the plane, i.e., we map $P$ to $\R^2$.

\begin{problem}
Given an arrangement $\mathcal{L}$ of $n$ lines in $\R^3$ in general position colored red, green, and blue, does there exist a cross-section determined by a vertical plane for which there exists a line that simultaneously bisects all three classes of points?
\end{problem}

Recently, Schnider~\cite{patrick_socg} was able to answer this question in the affirmative.
Actually, the plane defining such a cross-section can be chosen to contain any vertical line disjoint from~$\mathcal{L}$.
However, this freedom of choice can not be used to show an analogous statement for four colors, Barba showed that three colors is tight if we restrict ourselves to vertical planes.%
\footnote{Luis Barba, personal communication (2018).}
On the other hand, if we allow all possible planes, a stronger statement might be possible, but different arguments than the ones in~\cite{patrick_socg} would be necessary.
We will reprove the following version of the result for three colors:

\begin{theorem}\label{thm:main}
Given an arrangement $\mathcal{L}$ of $n$ lines in $\R^3$ in general position colored red, green, and blue, none of which intersect the $z$-axis.
Then, there exists a plane through the $z$-axis that defines a cross-section for which there exists a line that simultaneously bisects all three classes of points.
\end{theorem}

Note that in~\cite{patrick_socg}, the result is phrased differently, but this version follows analogously from the arguments in that paper.
The proof in~\cite{patrick_socg} is existential using topological methods.
In this work, we provide a combinatorial proof (i.e., a proof that uses only methods from discrete geometry and combinatorics) of Theorem~\ref{thm:main} that sheds further light on the structure of the problem.
This allows us to devise an $O(n^{2}\log^{2}(n))$ time algorithm to find such a plane and the bisector of the induced cross-section.

\subparagraph*{Preliminaries.}
Every triple of points of a point set in the plane is oriented either clockwise or counterclockwise, or it is collinear.
Two sets of $n$ points (not all on a single line) with a given bijection between them have the same \emph{order type} if each triple in one set is oriented in the same way as its image in the other set.
The orientation of a point triple $(p,q,r)$ is given by the sign of the determinant
\[
\begin{vmatrix}
x_p & x_q & x_r \\
y_p & y_q & y_r \\
1 & 1 & 1
\end{vmatrix} \enspace .
\]
In particular, it is $0$ if the points are on a line.
Whether there is a bisecting line for a cross-section only depends on its order type (rather than on the actual coordinates of the individual points).
We first observe that the points of a cross-section belonging to three lines become collinear only a constant number of times during the rotation of~$P$.
W.l.o.g., let $\vec{n} = (t, 1, 0)^T$ be the normal vector of $P$ at time~$t$.
Then the coordinates of the intersection point of a line and $P$ is a function in~$t$.
Three points become collinear if the above determinant is zero.
Since we have a rational function of constant degree, the number of times at which three points of a cross-section become collinear is constant.
Further, there is a change in the order type if one of the lines is parallel to $P$, which happens only once.
In total, such a cross-section of $n$ lines thus determines $O(n^3)$ different order types.
(See \cite{lund,crosssections_journal} for similar observations regarding translated planes.)
A brute-force algorithm would therefore test $O(n^3)$ cross-sections for a bisecting line.
(We later show that such a test can be done in $O(n^{{4}/{3}}\log^{1+\epsilon}(n)$ time.)
Our algorithm thus significantly improves over this trivial time bound.

\subparagraph*{Related work.}
There is a vast number of algorithmic results for problems related to ham-sandwich cuts.
Lo, Matou\v{s}ek, and Steiger~\cite{lo} provide a linear-time algorithm for finding a ham-sandwich cut in the plane (following a linear-time randomized algorithm by Lo and Steiger~\cite{lo_steiger}).
They also provide algorithms for higher dimensions.
In contrast to that, optimally counting the number of ham-sandwich cuts of a bi-colored point set is still an open problem.
The current best algorithm relies on walking the median level of the dual line arrangement and keeping track of the levels of the two color classes~\cite{weighted_ham_sandwich}.
This can be done in $O(n^{4/3}\log^{1+\epsilon}(n))$ time~\cite{Chan, constructing_belts} for any $\epsilon > 0$.
(We will also use this as a sub-algorithm.)

Arrangements of lines in $\R^3$ were probably first considered as a data structure by McKenna and O'Rourke~\cite{mckenna}.
Chazelle et al.~\cite{lines_in_space} consider further algorithms and properties related to such arrangements.
Lines in $\R^3$ naturally appear in connection with kinetic data structures.
Problems that involve points moving at individual constant speed can be modeled as problems on line arrangements in $\R^3$ and cross-sections with vertically translated planes (see, e.g.,~\cite{mobile_data} for a historical account).
For such problems, the number of changes in the combinatorial structure of the point set is crucial.
For example, the order type of a cross-section changes only at the time when three points become collinear.
Among other results, Lund et al.~\cite{lund} give a tight bound of $2\binom{n}{3}$ on the number of point triples that become collinear at some point in time.
Aichholzer et al.~\cite{crosssections_journal} show that there appear $O(n^9)$ different order types in that setting.
Combinatorial changes in minimum and maximum spanning trees~\cite{mst}, Voronoi diagrams~\cite{voronoi}, as well as convex hulls and bounding boxes~\cite{extent} of such moving points have also been investigated.

Let us emphasize that our setting is different from the one in which the plane is translated.
When rotating a plane around the $z$-axis, it will at some point be parallel to each line of $\mathcal{L}$, making the corresponding point vanish in the cross-section.
This does not happen for translated planes.
Also, our trajectories are not linear.
Lund et al.~\cite{lund} give an example of a line arrangement in which no point triple changes its orientation;
in our setting, after rotating the plane around the $z$-axis by $180^\circ$, the orientation of all point triples in the cross-section has changed.

\section{A Combinatorial Proof}
\label{sec:proof}

In this section we give an alternative proof of Theorem \ref{thm:main}.
The proof relies on an analysis of the changes in the combinatorial structure of the cross-sections when rotating a plane containing the $z$-axis.
In particular, we will show that during a rotation of the plane by $180^{\circ}$, one of the cross-sections must be such that there exists a line that simultaneously bisects all three classes of points.
Further, the method used to prove this allows us to algorithmically decide, for some moment in the rotation, whether a solution appears before or after this moment.
This sidedness oracle is a crucial element for the algorithm described in Section~\ref{sec:algorithm}.

\subsection{Dual interpretation in cross-sections}

For every directed plane $P$ containing the $z$-axis, denote by $S(P)$ the induced cross-section.
Consider the dual line arrangement $A(P)$, under the standard point-line-duality.
(For example, a point $(x_p, y_p)$ is mapped to the line with equation $y = x_px+y_p$.)
For an arrangement of $m$ lines, the \emph{$k$-level} is the set of points that has at most $k$ lines below it and at most $m-k-1$ lines above it.
If there is no vertical line, the $k$-level is a piece-wise linear $x$-monotone curve.
The $\floor{\frac{m}{2}}$-level is called the \emph{median level}.
Denote by $r(P)$, $b(P)$ and $g(P)$ the median levels of the red, green and blue lines in $A(P)$, respectively.
A line that bisects all three point sets in a cross-section $S(P)$ corresponds to an intersection point of the three median levels $r(P)$, $b(P)$ and $g(P)$.
Thus, we want to show that there exists some plane $P$ such that in the induced line arrangement $A(P)$ the three median levels $r(P)$, $b(P)$ and $g(P)$ have a common intersection.

To this end, we start with some plane $P_0$ through the $z$-axis and rotate it around that axis, keeping track of the changes in the induced arrangement.
After a rotation of angle~$\pi$, we will arrive again at $P_0$, but oriented the other way;
we will call this oriented plane $P_1$.
Note that $S(P_1)$ is $S(P_0)$ mirrored at the $z$-axis.
Similarly, $A(P_1)$ is $A(P_0)$ mirrored at the $z$-axis.
(The first coordinate of a point corresponds to the slope of the dual line; hence, mirroring along the $z$-axis, which is the $y$-axis of the plane, corresponds to inverting all slopes, which in turn is equivalent to mirroring along
the $z$-axis.)
We will now study what type of changes can occur in the intersection pattern of the middle levels during this rotation.

For simplicity, we assume that each color class of $\mathcal{L}$ contains an odd number of lines.
(This assumption may be dropped at the cost of a more involved exposition.)
Then all three median levels are $x$-monotone curves along lines of the arrangement, except in the degenerate cases where the plane is parallel to one of the lines.
We will look at these degenerate cases later.
For now, suppose no line is parallel to $P$.
We walk along the green median level $g(P)$ (which is now just a curve) and look at the order of intersections with $r(P)$ and $b(P)$ from left to right.
We say that a level $m$ intersects $g(P)$ from  \emph{below} if $m$ is below $g(P)$ before the intersection and above afterwards.
Similarly, $m$ can intersect $g(P)$ from \emph{above}.
Note that if an intersection of $m$ and $g(P)$ is from below, then the next intersection has to be from above.
We will call intersections of $r(P)$ with $g(P)$ \emph{red} intersections and intersections of $b(P)$ with $g(P)$ \emph{blue} intersections.
Further note that when reorienting $P$, the order of intersections reverses, and all intersections from above become intersections from below and vice versa.

When rotating the plane $P$, the dual line arrangement and thus the median levels $r(P)$, $b(P)$, and $g(P)$ change continuously as long as no line becomes parallel to~$P$.
Note that the order of intersection along $g(P)$ can only change when the order type of $S(P)$ (and thus also $A(P)$) changes.
Further, the only changes in the order type of $A(P)$ that affect the order of intersection along $g(P)$ are either the degenerate cases or the ones where a vertex of $A(P)$ passes over a segment of $g(P)$.
Thus, there are the four ways in which the order of intersection along $g(P)$ can change during a rotation of~$P$.
\begin{enumerate}
\item The first is that a red intersection and a blue intersection switch places (i.e., the relative order in which they appear along $g(P)$).
In this case, the two intersections coincide at some point in time, and the point at which they coincide is an intersection that we are looking for;
in particular, this means that three lines of $\A(P)$ meet in a single point.
See Figure \ref{Fig:Switch} for an illustration.
\item The second case happens when (without loss of generality) two red intersections, the first from below, the second from above, coincide.
Then, rotating further, the two intersections vanish.
This type of change does not indicate a solution.
See Figure \ref{Fig:Deletion} for an illustration.
\item The third case happens when the last segments of $g(P)$ and (without loss of generality) $r(P)$ become parallel.
Then the last red intersection vanishes, but it reappears as the first red intersection as the rotation of~$P$ continues.
Furthermore, if this intersection was from above, it is now from below, and vice versa.
Also this type of change will not give us a solution.
See Figure \ref{Fig:Transfer} for an illustration.
(This case does not give a change when considering $P$ as a projective plane;
however, as we prefer to have an order on the intersections along $g(P)$, we consider $P$ to be Euclidean.)
\item Finally, a last type of change occurs in the degenerate case when the plane becomes parallel to one of the lines;
call it $\ell$ and assume without loss of generality that it is red.
See Figure~\ref{Fig:Degenerate} for an illustration.
Let $m$ be the number of red lines.
In this case, the point of intersection of $\ell$ and $P$ moves towards infinity, vanishes when $P$ is parallel to $\ell$, and then enters from the other side.
In the dual, this corresponds to a line $h$ getting steeper and steeper, until it vanishes when $P$ is parallel to $\ell$, and then rotating further, changing the sign of its slope.
(It is useful to consider $h$ to become vertical at this limit instead of vanishing.)
In this situation, the median level $r(P)$ changes quite drastically: assume without loss of generality that $h$ has negative slope before it becomes orthogonal and positive slope afterwards.
Then, shortly before $h$ vanishes, it has the smallest slope of all red lines.
Afterwards, it has the largest slope of all red lines.
Furthermore, the new median level (i.e., $\floor{\frac{m}{2}}$-level) of the red arrangement is the former $(\floor{\frac{m}{2}}-1)$-level to the left of the intersection of $h$ and $r(P)$, and it is the former $(\floor{\frac{m}{2}}+1)$-level to the right of this intersection.
In particular, several pairs of red intersections might have vanished and several red and blue intersections might have switched places.
However, at the moment when $h$ is vertical (or actually vanishes), it plays no role for the median level.
Thus, the red median level is actually the whole area between the former $(\floor{\frac{m}{2}}-1)$-level and the former $\floor{\frac{m}{2}}$-level to the left of $h$, and the area between the former $\floor{\frac{n}{2}}$-level and the former $(\floor{\frac{m}{2}}+1)$-level to the right of $h$.
Thus, if a red and a blue intersection switch places, the blue intersection lies inside the red median level at the moment when $h$ is vertical, and we again have a solution.
(In $S(P)$, this intersection corresponds to a line through a green and a blue point, bisecting an even number of red points.)
The same argument holds if $\ell$ is blue and a similar argument applies if $\ell$ is green.
\end{enumerate}

\begin{figure}
\centering
\includegraphics[scale=0.5]{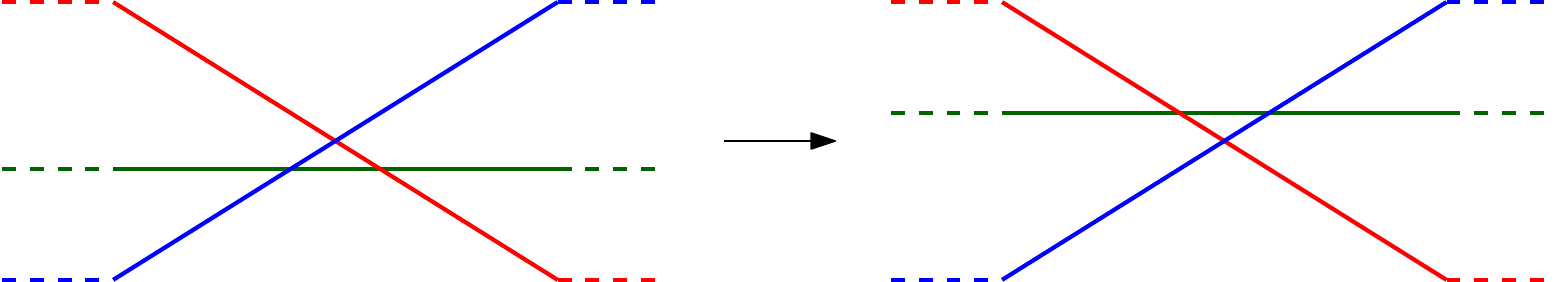}
\caption{A red intersection and a blue intersection switching places.}
\label{Fig:Switch}
\end{figure}

\begin{figure}
\centering
\includegraphics[scale=0.5]{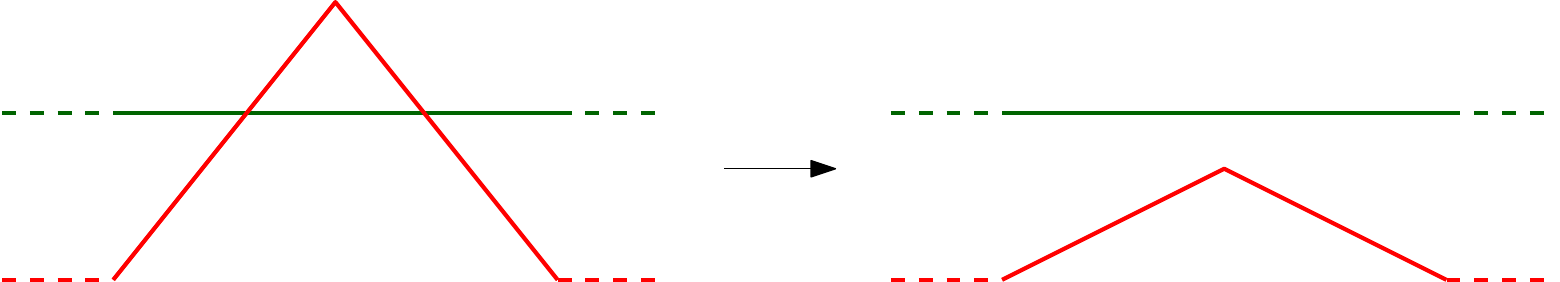}
\caption{Two red intersections vanishing.}
\label{Fig:Deletion}
\end{figure}

\begin{figure}
\centering
\includegraphics[scale=0.5]{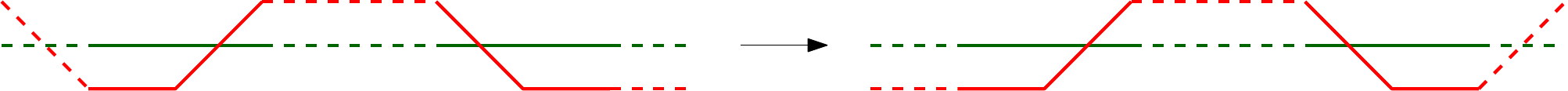}
\caption{The first red intersection vanishing and reappearing as the last red intersection.}
\label{Fig:Transfer}
\end{figure}

\begin{figure}
\centering
\includegraphics[scale=0.5]{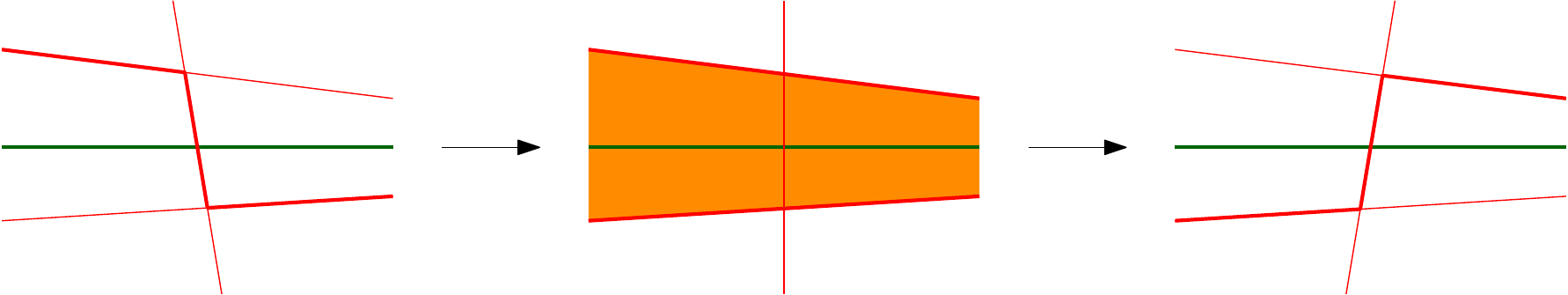}
\caption{A red line becoming vertical.}
\label{Fig:Degenerate}
\end{figure}

In conclusion, we can track the intersections of $r(P)$ and $b(P)$ along $g(P)$ and we find a solution whenever a red intersection and a blue intersection switch places.
Furthermore, we know that after rotating all the way from $P_0$ to $P_1$ the order of intersections has reversed and all intersections that were from above are now from below and vice versa.
Thus, what we want to show is that at some point during the rotation, a red and a blue intersection need to switch places.
We will do this in a purely combinatorial setting in the next section.

\subsection{Bi-chromatic Sign Sequences}

We define a \emph{bi-chromatic sign sequence} $S$ as a finite sequence of signs $s\in\{+,-\}$ that satisfies the following conditions:
\begin{enumerate}
\item[(1)] each symbol $s$ is either red or blue;
\item[(2)] the subsequence of all symbols of one color has odd length and alternates between $+$ and~$-$.
\end{enumerate}

Note that the red and blue intersections along $g(P)$ define a bi-chromatic sign sequence when writing intersections from above as $+$ and intersections from below as $-$.
See Figure \ref{Fig:Sequence} for an illustration.

\begin{figure}
\centering
\includegraphics[scale=1]{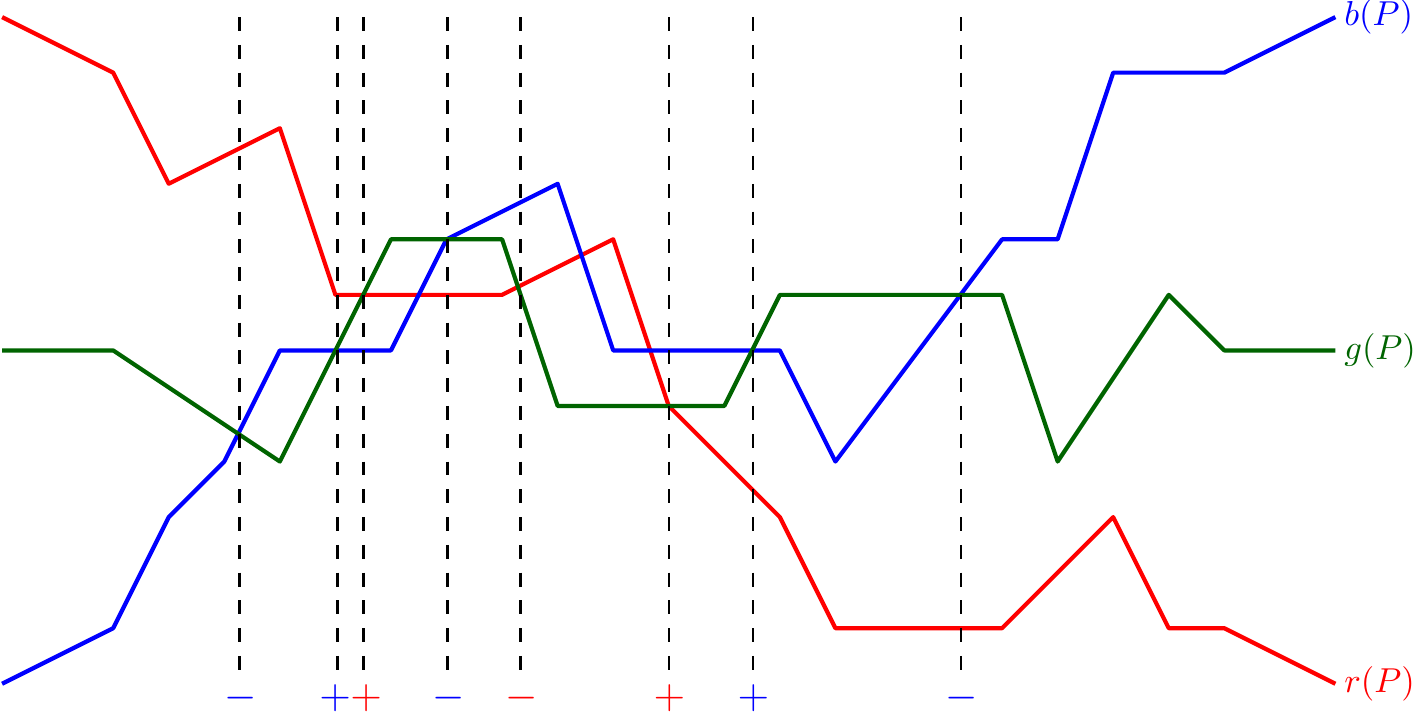}
\caption{Three median levels and their corresponding bi-chromatic sign sequence.}
\label{Fig:Sequence}
\end{figure}

We define the \emph{inverse} of a sign $s$, denoted by $\overline{s}$ as the other sign, but in the same color.
We further define the \emph{reverse} of $S=s_1,\ldots, s_k$, denoted by $\overline{S}$, as $\overline{s_k},\ldots,\overline{s_1}$.
Note that if $S$ is the sequence of red and blue intersections for some $P$, then reorienting $P$ gives rise to $\overline{S}$.

We consider the following operations on bi-chromatic sign sequences:
\begin{enumerate}
\item \textbf{Switch:} Let $s_i$ and $s_{i+1}$ be two consecutive elements of $s_1,\ldots, s_k$ of different color. Then we can switch $s_i$ and $s_{i+1}$ to get a new sequence $s_1,\ldots, s_{i-1},s_{i+1},s_i,\ldots, s_k$.
\item \textbf{Insertion/Deletion:} Let $s_i$ and $s_{i+1}$ be two consecutive elements of $s_1,\ldots, s_k$ of the same color.
(Hence, their signs are different.)
Then we can delete $s_i$ and $s_{i+1}$ to get a new sequence $s_1,\ldots, s_{i-1},s_{i+2},\ldots, s_k$.
Similarly, let either (i) $s_i$ and $s_{j}$ be two elements of the same color that are consecutive in the subsequence of their color, or (ii) $s_j$ be the first element of its color in $S$ and $s_j=\overline{s_i}$ or (iii) $s_i$ be the last element of its color in $S$ and $s_i=\overline{s_j}$.
Then we can insert $\overline{s_i}$ and $\overline{s_j}$ consecutively (w.r.t.\ the entire sequence) somewhere after $s_i$ before $s_j$ to get a new sequence $s_1,\ldots, s_i,\ldots, \overline{s_i}, \overline{s_j}, \ldots, s_j,\ldots, s_k$.
(Case (ii) and (iii) correspond to an insertion at the beginning or the end of $S$, respectively.)
\item \textbf{Transfer:} Let $s_k$ be the last element of the sequence. Then we can transfer $s_k$ to the beginning of the sequence to get a new sequence $\overline{s_k}, s_1,\ldots, s_{k-1}$. Similarly, we can transfer the first element $s_1$ to the end of the sequence to get a new sequence $s_2,\ldots, s_k,\overline{s_1}$.
\end{enumerate}

These operations describe exactly the changes that can happen for a sequence of red and blue intersections: Case 1 translates to a switch, case 2 gives an insertion or a deletion, case 3 corresponds to a transfer.
For case 4, we need the following lemma:

\begin{lemma}\label{case4}
Consider the degenerate case when the plane becomes parallel to one of the lines.
Assume that there is no solution at this point of rotation.
Then the change in the bi-chromatic sign sequence defined by the red and blue intersections along $g(P)$ can be described as a sequence of insertions, deletions and transfers.
\end{lemma}

\begin{proof}
Assume without loss of generality that the line $\ell$ that the plane becomes parallel to is red.
In the dual of the cross-section, we will, as above, consider the dual line $h$ to become vertical.
Assume again without loss of generality that $h$ has negative slope before it becomes orthogonal and positive slope afterwards.
Let $r_0(P)$ and $r_1(P)$ be the red median levels right before and just after $h$ is vertical.
From the assumption that there is no solution at this point of rotation, we conclude that all intersections of $b(P)$ with $g(P)$ must be either below or above both $r_0(P)$ and $r_1(P)$.
In particular, the only changes that can happen between two consecutive intersections of $b(P)$ with $g(P)$ are insertions and deletions.
It remains to analyze the changes before the first and after the last intersection of $b(P)$ with $g(P)$.
For this, we distinguish two cases.
In the first case, both $r_0(P)$ and $r_1(P)$ start on the same side of $g(P)$, without loss of generality above.
Note that then both $r_0(P)$ and $r_1(P)$ need to end below $g(P)$.
In particular, the same argument as above also holds before the first and after the last intersection of $b(P)$ with $g(P)$.
In the second case, $r_0(P)$ starts above $g(P)$, whereas $r_1(P)$ starts below $g(P)$.
Then, $r_0(P)$ ends below $g(P)$ and $r_1(P)$ ends above $g(P)$.
In particular, without loss of generality, the bi-chromatic sign sequence for $r_0(P)$ starts with a red $+$, while the bi-chromatic sign sequence for $r_1(P)$ ends with a red $-$.
This change can be described by a transfer.
Again, by the above arguments, all other changes correspond to insertions and deletions.
The analogous proof holds when the line $\ell$ that the plane becomes parallel to is blue and similar arguments apply when $\ell$ is green.
\end{proof}

In fact, it can be shown that case 4 can always be described as a sequence of insertions, deletions, transfers and switches, but the above lemma is enough for our purposes.
Let us summarize our findings so far in a lemma:

\begin{lemma}\label{lem:summary}
Let $\mathcal{L}$ be an arrangement of $n$ lines in $\R^3$ colored red, green, and blue, none of which intersect the $z$-axis.
For every oriented plane $P$ containing the $z$-axis, the induced cross-section defines a bi-chromatic sign sequence.
Further, if $P_0$ and $P_1$ are two oriented planes and $S_0$ and $S_1$ are their induced bi-chromatic sign sequences, rotating $P_0$ until it coincides with $P_1$ gives rise to a sequence of bi-chromatic sign sequences from $S_0$ to $S_1$, where two consecutive sign sequences differ by a switch, insertion, deletion or transfer.
Finally, every switch used in this transformation from $S_0$ to $S_1$ corresponds to a cross-section for which there exists a line that simultaneously bisects all three classes of points.
\end{lemma}

We will show that it is not possible to transform a sequence $S$ to $\overline{S}$ using only insertions, deletions and transfers.
To this end, we first define the \emph{reduced sign sequence} of a bi-chromatic sign sequence: look at a subsequence of consecutive elements of the same color.
If this subsequence has even length, remove it.
Iteratively removing such consecutive mono-chromatic subsequences of even length leaves us with a sequence that alternates between the colors.
Further, it is easy to check that this sequence is still a bi-chromatic sign sequence.
Finally, as both the subsequence of red elements as well as the subsequence of blue elements have odd length, the new sequence is not empty and it starts with a different color than it ends with (a fact that will be crucial in proving the following lemmas).
We call the sequence obtained this way the reduced sign sequence of a bi-chromatic sign sequence $S$, and denote it by $\rho(S)$.

At last, we arrive at our final definition.
For a reduced sign sequence $\rho(S)$, we consider two functions.
The first, denoted by $\Sigma(\rho(S))$, is the sum of ``$+$''-symbols minus the number of ``$-$''-symbols in $\rho(S)$.
Note that the value of this function is always in $\{-2,0,2\}$.
The second function, denoted by $\lambda(\rho(S))$, is $2$ whenever the first symbol of $\rho(S)$ is red and $0$ otherwise.
We now define the \emph{trace} $\tau$ of a bi-chromatic sign sequence as follows:
\[
\tau(S):=\frac{\Sigma(\rho(S))+\lambda(\rho(S))}{2}\text{ mod }2 \enspace .
\]

We continue by proving a few crucial lemmas about traces.

\begin{lemma}\label{lem:trace_operations}
The trace of a bi-chromatic sign sequence is invariant under insertions, deletions and transfers.
\end{lemma}

\begin{proof}
It follows immediately from the definitions that the reduced sign sequences, and thus the trace, is invariant under insertions and deletions.
As for transfers, consider the two sequences $S=s_1,\ldots,s_k$ and $S'=s_2,\ldots,s_k,\overline{s_1}$.
If their reduced sign sequences are equal, so are their traces, so assume that their reduced sign sequences are different.
In this case we claim that we either have (a) $\rho(S)=r_1,r_2,\ldots, r_m$ and $\rho(S')=r_2,\ldots, r_m,\overline{r_1}$ or (b) $\rho(S)=r_1,r_2,\ldots, r_m$ and $\rho(S')=\overline{r_m},r_1,\ldots, r_{m-1}$.
Assume without loss of generality that $r_1$ is red, and thus $r_m$ is blue.
We will distinguish the cases where $s_1$ is red or blue.
If $s_1$ is red we can assume without loss of generality that $s_1=r_1$.
In particular, $\overline{s_1}$ will be the last element of $\rho(S')$.
Further, as all blocks of the same color in $S$ except the first and last one do not change their length with a transfer, all other elements in $\rho(S)$ stay the same and we are in case (a).
If $s_1$ is blue, the transfer will either create or destroy an odd-length blue block in the beginning of $S$.
In particular, $\rho(S')$ will begin with a blue element $r_0$.
On the other hand, the transfer will either create or destroy an even-length blue block at the end of $S$ and thus $\rho(S')$ will end with the red element $r_{m-1}$.
Since $\rho(S')$ is still a bi-chromatic sign sequence, we conclude that $r_0=\overline{r_m}$ and we are in case (b).
In both cases we deduce that $\Sigma(\rho(S'))=\Sigma(\rho(S))\pm 2$ and $\lambda(\rho(S'))=\lambda(\rho(S))\pm 2$.
Thus $\tau(S'):=\frac{1}{2}(\Sigma(\rho(S))+\lambda(\rho(S))\pm 4)\text{ mod }2 = \tau(S)\pm 2\text{ mod }2=\tau(S)$.
\end{proof}

Note that the above in general does not hold for switch operations.
While there is no change in the sum of the signs in the sign sequence by a switch operation, this does not hold for the reduced sign sequence, as the parities of adjacent substrings of the same color changes.

\begin{lemma}\label{lem:trace_reverse}
The traces of a bi-chromatic sign sequence $S$ and its reverse $\overline{S}$ are different.
\end{lemma}
\begin{proof}
As in the reverse of a sequence all signs change, we have $\Sigma(\rho(\overline{S}))=-\Sigma(\rho(S))$.
Further, the first color in $\rho(\overline{S})$ is the last color in $\rho(S)$ and thus different from the first color of $\rho(S)$.
Hence, we have $\lambda(\rho(\overline{S}))=\lambda(\rho(S))\pm 2$.
Plugging this into the definition of $\tau$, we conclude that $\tau(\overline{S})\neq\tau(S)$.
\end{proof}

Combining Lemma \ref{lem:trace_operations} and Lemma \ref{lem:trace_reverse}, we conclude the following:

\begin{corollary}\label{cor:switches_needed}
Any sequence of operations transforming a bi-chromatic sign sequence $S$ into its reverse $\overline{S}$ contains at least one switch.
\end{corollary}

Combining this with Lemma \ref{lem:summary}, we obtain a proof of Theorem \ref{thm:main}.
In particular, knowing the traces of $S$ and $\overline{S}$, computing the trace for a cross-section between $S$ and $\overline{S}$ also tells us in which direction a switch has to occur.
In order to have a sidedness oracle for a plane in our original problem, we thus need to efficiently compute the trace of the reduced sign sequence of the cross-section determined by this plane.

\begin{lemma}\label{lem:obtain_sign_sequence}
\sloppypar{
The reduced sign sequence of a cross-section of $n$ lines can be obtained in $O(n^{4/3}\log^{1+\epsilon}(n))$ time.
}
\end{lemma}

\begin{proof}
We first compute each color class of the bi-chromatic sign sequence separately.
To this end, we walk the median level of the line arrangement of the two involved color classes, keeping track of the levels of them (note that every intersection between the two median levels lies in the median level of their total arrangement).
That is, we walk the median level of the red and the green lines, and then the median level of the blue and the green lines.
This can be done in $O(m\log^{1+\epsilon}(n))$ time~\cite{Chan, constructing_belts}, where $m$ is the complexity of the median level.
The current best bound for the complexity of a median level of an arrangement of $n$ lines is $O(n^{4/3})$~\cite{Dey}.
Further, the above argument shows that the number of elements of each color class in the bi-chromatic sign sequence is bounded by the complexity of the corresponding median level of the arrangement of two colors.
Given the two subsequences, we can easily merge them to the order in which they appear from left to right into the bi-chromatic sign sequence.
The time this requires is linear in the length of that sequence, which we have just observed to be in $O(n^{4/3})$.
Finally, given a bi-chromatic sign sequence, its reduced sign sequence can easily computed in time linear in the length of the bi-chromatic sign sequence.
\end{proof}

All that remains now for the sidedness oracle is to compute the trace of the sign sequence, which can again be easily done in time linear in the length of the reduced sign sequence.

\begin{corollary}\label{cor:side}
Given a plane through the $z$-axis, we can determine on which side of the plane there exists bisecting cross-section in $O(n^{4/3} \log^{1+\epsilon}(n))$.
\end{corollary}

\section{The Algorithm}
\label{sec:algorithm}
We may now use the above insight to devise an algorithm to find the defining plane of a cross-section with a simultaneously bisecting line.
We first recall that the existence of a bisecting line depends only on the order type of the cross-section.
As all events that change the order type of a cross-section between $P_a$ and $P_b$ are the ones in which three points of the cross-section become collinear, we need to find such an event.
As noted in the introduction, within that interval, any triple of points in the cross-sections can change its orientation only a constant number of times, during which it becomes collinear.
As there can be $\Theta(n^3)$ such events, checking all of them would already give a polynomial time algorithm, but we want to do something faster.

\subsection{A general framework}
We will now develop an algorithmic framework that allows us to to find the defining plane of a cross-section with a simultaneously bisecting line efficiently.
As all events that change the order type of a cross-section between $P_a$ and $P_b$ are the ones in which three points of the cross-section become collinear, we need to find such an event.
As noted in the introduction, within that interval, any triple of points in the cross-sections can change its orientation only a constant number of times, during which it becomes collinear.
As there can be $\Theta(n^3)$ such events, checking all of them would already give a polynomial time algorithm, but we want to do something faster.
For our framework we will assume the following input:

\begin{enumerate}
\item a continuously changing family of point sets $P_t, t\in[0,1]$ of size $n$ with the property that each triple of points changes their orientation only a constant number of times and for each triple these events can be found in constant time. This family contains some solution point sets $P_x$.
\item a sidedness oracle $S$ which for every $t\in[0,1]$ computes in time $T(n)$ a value $S(t)\in\{-1,0,1\}$ with the following properties:
	\begin{enumerate}
	\item $S(0)=-1$;
	\item $S(1)=1$;
	\item if $S(t)=0$ then $P_t$ is a solution point set;
	\item for every $a<b$ such that $S(a)=-1$ and $S(b)=1$ or vice versa, the interval $[a,b]$ contains a solution point set.
	\end{enumerate}
\end{enumerate}

For example, the family of point sets could be defined by points in two colors, red and blue, moving at constant speed along trajectories that are described by polynomials of constant degree, the blue points starting left and ending right of the $y$-axis, the red points starting right and ending left of the $y$-axis.
The solution point sets could then be all point sets where both the blue and the red Tukey median regions can be intersected by a single vertical line.
A sidedness oracle could then be constructed by computing both Tukey median regions and checking if the are separated by a vertical line, returning $0$ if not, or $-1$ or $1$, depending on which color is to the left of the vertical line.
The runtime of the sidedness oracle in this example would be $O(n\log^3(n))$, see \cite{Tukey_median}.

We will use Megiddo's~\cite{megiddo} parametric search technique in the following way.
Let $P_a$ and $P_b$ be two point sets with $S(a)=-1$ and $S(b)=1$ or vice versa (we may start with $P_a = P_0$ and $P_b = P_1$).
We therefore know that there is at least one solution point set $P_x$ between them.
The high-level idea is to construct a representation of the order type of the (unknown) point set $P_x$ by simulating a parallel sorting algorithm for each point of $P_x$.
In each iteration, the sorting algorithm produces a batch of $O(n)$ comparisons that need to be answered.
Each comparison is related to a constant number of events.
There is a total order in which these events occur when going from $P_a$ to $P_b$.
By selecting the median of these events in time $O(n)$, we can answer approximately half of the events in each batch by applying the sidedness oracle once.
We repeat this $O(\log(n))$ times to answer all the events of the batch.
A batch therefore can be processed in $O((n+T(n)) \log(n))$ time.
In addition, we update $P_a$ and $P_b$ at each application of the sidedness oracle to narrow down the interval containing our solution $P_x$.
For one run of the sorting algorithm (requiring $O(\log(n))$ batches), this amounts to $O((n+T(n)) \log^2(n))$ time.
In total, we require $n$ runs of the sorting algorithm, giving the total running time of $O((n+T(n)) n \log^2(n))$.

After this high-level exposition, we now turn to the detailed description of our algorithm, where we will also describe how to further speed it up.

\paragraph{Order types and local sequences.}
The combinatorial representation of the order type we will use is the set of local sequences of unordered switches (for short, \emph{local sequences}), devised by Goodman and Pollack~\cite{semispaces}.
(See also the work by Streinu~\cite{streinu} and by Felsner and Weil~\cite{felsner_weil} on that topic.)
For any point~$p$ of a point set~$S$, consider the rotating line with pivot~$p$;
the order in which this line passes over the other points of $S$ is the local sequence of $p$ in $S$.
(While the local sequence is cyclic, it can be linearized by letting it start at the element of $S \setminus \{p\}$ with, say, the smallest label.)
The local sequences of all the points of $S$ determine the order type of~$S$~\cite{semispaces}.
We construct these sequences by considering $n$ different (but not independent) sorting problems:
for each point $p \in S$, sort the points of $S \setminus \{p\}$ around $p$ as they appear along the local sequence.

\paragraph{Determining the order type on $P_x$.}
Let $P=P_x$ be the unknown solution point set.
For each point $p \in P$, we run a parallel sorting algorithm that sorts the points of $P\setminus \{p\}$ according to the unsigned local sequence of~$p$.
Pick an arbitrary point $q \in P \setminus \{p\}$ with which the sequence should start (recall that it is actually circular).
This sorting algorithm has $O(\log(n))$ iterations.
In each iteration, we get a batch $B$ of $(n-1)/2$ point pairs of $S \setminus \{p\}$;
such a pair $(a,b)$ corresponds to a comparison of $a$ and $b$ in the order around $p$, that can be answered by knowing the order type of the four points $p,q,a,b$.
Between $P_a$ and $P_b$, there are at most a constant number of events at which this order type changes (all of which can be obtained in constant time).
Thus, for the batch $B$ of comparisons, we get a list $E$ of such events, whose size is in $O(|B|) = O(n)$.
Each event is associated with a value $t$ in the interval $[a,b]$.
As soon as we know for each event the side on which the solution $P$ lies, we can answer all the comparisons of the batch~$B$.
We pick the median of these events in $O(n)$ time, which gives us a point set $P'$ at which the event happens.
We apply the sidedness oracle to determine the side of $P'$ containing $P$ (or find a solution, if the oracle returns $0$).
Also, we update $P_a$ or $P_b$ by setting one to $P'$, maintaining the invariant that $P$ is between $P_a$ and $P_b$.
Hence, in $T(n) + O(n)$ time, we can process approximately half of the events;
after $O(\log(n))$ iterations of picking the median of the remaining events and narrowing the interval between $P_a$ and $P_b$, all events~$E$ of the batch $B$ have been answered.
That is, we are left with an interval between $P_a$ and $P_b$ that contains the solution $P$ and within which the orientation of the point triples involved in the comparisons does not change.
In total, the parallel sorting algorithm requires us to answer $O(\log(n))$ batches to sort the points around $p$.
Hence, after $O((n+T(n)) \log^2(n))$ time, the local sequence of a point $p$ in the point set $P$ is determined.
We do this for all $n$ points of $P$, further narrowing down the interval between $P_a$ and $P_b$.
Thus after a total of $O((n+T(n))n\log^2(n))=O(n^2\log^2(n)+T(n)n\log^2(n))$ steps, the local sequence of all the points and thus the order type of~$P$ is determined.
Further, in the final interval $[a,b]$, this order type does not change, meaning that a solution is found.

\paragraph{Speeding up.}
As an anonymous reviewer noted, the above algorithm can be sped up by using a slight variation of the parallel sorting algorithm: by bundling together $O(n)$ batches of all $n$ runs of the sorting algorithm, we get $O(\log(n))$ batches, each of size $O(n^2)$.
For each such batch, we pick the median of the events in $O(n^2)$ time and apply the sidedness oracle to determine the side of $P'$ containing $P$ (or find a solution, if the oracle returns $0$).
Hence, in $T(n) + O(n^2)$ time, we can process approximately half of the events, leaving us again with $O(\log(n))$ iterations to answer all events of the batch.
As we again have $O(\log(n))$ batches, the local sequence of all the points and thus the order type of~$P$ can thus be determined in time $O((n^2+T(n))\log^2(n))=O(n^2\log^2(n)+T(n)\log^2(n))$.

This algorithm proves the following:

\begin{theorem}\label{thm:framework}
Given an input to the framework described above, a solution point set $P$ can be found in time $O(n^2\log^2(n)+T(n)\log^2(n))$.
\end{theorem}

\subsection{Applying the framework}
We may now use the above framework to devise an algorithm to find the defining plane of a cross-section with a simultaneously bisecting line.
We first recall that the existence of a bisecting line depends only on the order type of the cross-section.
In this terminology, we observe that there exists a bisection only if three points of the cross-section become collinear (i.e., when the median levels of the three classes intersect in a single point), or when a line of $\mathcal{L}$ becomes parallel to the defining plane~$P$ (and thus the intersection of two median levels is inside the two-dimensional region defined by the third one).
In particular, the latter events constitute a change in the order type that is not captured by our framework.
When rotating the plane $P_0$ by an angle of $\pi$, there are $n$ events in which a line becomes collinear to the plane.
We will now describe how we handle these events.

We start by arbitrarily choosing $P_0$ (and thus $P_1$).
For each line in $\mathcal{L}$, we can determine in constant time the plane between $P_0$ and $P_1$ that is parallel to it.
For each of these $n$ planes, we would like to compute the reduced sign sequence.
If the solution is at one of these $n$ events, we are done.
Otherwise, the sign sequences give us two planes between which there is a solution, and between which there is no event in which a line becomes parallel to the rotating plane.
This interval (or a solution) can be found using binary search in $O(n \log(n) + T(n)\log(n))$ time.
We let $P_a$ and $P_b$ be these two planes.

From the observations in the introduction, it follows that the number of the other events is in $O(n^3)$;
in particular, recall that the number of times a point triple on $P$ changes its order type is constant.
Further, Corollary~\ref{cor:side} gives us a sidedness oracle for these events.
We can thus apply to above framework with $P_0=P_a$ and $P_1=P_b$.
The runtime of the framework dominates the above computations, thus the total runtime of our algorithm is $O(n^2\log^2(n)+T(n)\log^2(n))$.
Recall that from Corollary~\ref{cor:side}, we have $T(n)\in O(n^{4/3} \log^{1+\epsilon}(n))$, hence we get the following:

\begin{theorem}
Given an arrangement $\mathcal{L}$ of $n$ lines in $\R^3$ colored red, green, and blue, none of which intersect the $z$-axis.
Then, a plane through the $z$-axis that defines a cross-section for which there exists a line that simultaneously bisects all three classes of points can be found in $O(n^2\log^2(n))$ time.
\end{theorem}

\section{Conclusion}
We have given a combinatorial proof of Theorem \ref{thm:main}.
The underlying structures that we discovered allowed us to give an $O(n^2\log^2(n))$ time algorithm.
We did this by providing a general framework that can be used for continuously changing point sets.
We do not believe that the runtime of our algorithm is best possible, it would be interesting to see whether other techniques might give a better runtime.
Applying techniques like $\varepsilon$-nets and cuttings so far did not lead to an improved algorithm.

It is also not clear that we really need all possible events: if one could for example show, that there is always a bisector at an event where one line becomes parallel to the plane, we would only need the first part of the algorithm, which would reduce its runtime.
We were so far not able to find an example of an arrangement where no solution is at such an event.

Apart from the existence of a ham-sandwich cut there are many other problems that could be studied for cross-sections of a rotating plane, e.g.\ finding a cross section has the point with the highest (Tukey, simplicial, ...) depth.
We expect that our framework can be used to solve some of these problems.

\subparagraph{Acknowledgments.}
We thank Luis Barba for bringing the problem to our attention. We also thank the anonymous reviewers for their helpful comments, in particular the reviewer who provided a significant improvement of the runtime of our algorithm.

\bibliographystyle{plainurl} %
\bibliography{refs}

\end{document}